\documentclass[letterpaper,12pt,peerreviewca,draftcls]{IEEEtran}
\usepackage{csm16}
\usepackage[margin=1in]{geometry}
\usepackage[nolists,nomarkers,tablesfirst]{endfloat} 
\usepackage{url}
\usepackage{physics}
\usepackage{commath}
\usepackage{amsmath}
\usepackage{bm}
\usepackage{tabularx}
\usepackage{graphicx,xcolor}
\usepackage{verbatim}
\usepackage{float,amsmath,amsfonts,amssymb,amsthm,color,graphicx,cite}
\usepackage{graphics} 
\usepackage{booktabs} 
\usepackage{caption, subcaption}
\usepackage{breqn}
\usepackage{array, mathtools}
\usepackage{pgfplots}
\usepackage{siunitx}
\usepackage{algorithm}
\usepackage[normalem]{ulem}
\usepackage[noend]{algpseudocode}
\usepackage{tabulary}
\usepackage{stmaryrd} 
\usepackage{varwidth}
\def\beq{\begin{equation}}
\def\eeq{\end{equation}}
\def\dsp{\displaystyle}
\def\reals{{\bf R}}

\usepackage{tikz}
\usepackage{pgfplots}
\usetikzlibrary{decorations.markings}
\usetikzlibrary{patterns}
\usetikzlibrary{arrows}
\usetikzlibrary{shapes}
\usetikzlibrary{fit}
 \usetikzlibrary{calc}
\usetikzlibrary{decorations.pathreplacing}
\usetikzlibrary{arrows,positioning} 
\usetikzlibrary{pgfplots.groupplots}
\usepackage{soul}
\usepackage{thmtools}

\declaretheorem[style=plain,name=Theorem]{theorem}

\usepackage{hyperref}
\hypersetup{
    colorlinks=true,
    linkcolor=black,
    filecolor=magenta,      
    urlcolor=black,
}

\makeatletter
\newcommand{\verbatimfont}[1]{\def\verbatim@font{#1}}%
\makeatother
\verbatimfont{\ttfamily\small}

\newcommand{\bi}{\begin{itemize}}\newcommand{\ei}{\end{itemize}}
\newcommand{\be}{\begin{equation}}\newcommand{\ee}{\end{equation}}
\newcommand{\bee}{\begin{enumerate}}\newcommand{\eee}{\end{enumerate}}
\newcommand{\bea}{\begin{eqnarray}}\newcommand{\eea}{\end{eqnarray}}
\newcommand{\beas}{\begin{eqnarray*}}\newcommand{\eeas}{\end{eqnarray*}}
\newcommand{\bc}{\begin{center}}\newcommand{\ec}{\end{center}}
\usepackage[left,pagewise]{lineno} 
\usepackage[english]{babel} 
\usepackage{blindtext}

\makeatletter
\renewcommand*\env@matrix[1][*\c@MaxMatrixCols c]{%
  \hskip -\arraycolsep
  \let\@ifnextchar\new@ifnextchar
  \array{#1}}
\makeatother

\title{From Observability to Observer Realization\\
\Large A path via elementary block-diagram manipulations}
\author{Eder Baron-Prada, \ Renzo Caballero, \ Eric Feron}

\begin{document}
\maketitle
\CSMsetup
\linenumbers \modulolinenumbers[2] 

Introductory state-space linear control courses focus on linear, time-invariant systems and spend intense efforts by introducing system realizations that allow the student to grasp fundamental concepts, among which controllability, observability, and controller and observer design. Assuming that the system is single-output, observability realizations allow the student to easily perceive the limitations of the observed output in terms of what it can estimate about the state of the system. Once observability is established, the observer realization provides a rich framework allowing the educator to establish the theorem stating the equivalence between observability and the ability to choose arbitrary pole locations for the corresponding state estimation error dynamics.  The pole location is achieved by introducing a linear gain which leads to a trivial observer pole placement process. While the observer realizations are trivial to obtain from systems expressed as transfer functions, such is not the case when starting from arbitrary observable state-space representations and elementary but laborious matrix manipulations transform the observability realization of a given system into its observer realization. 
This note describes a graphical mechanism to transform a system expressed in observability form into its observer form based on elementary block diagram manipulations, a technique whose pedagogical usefulness tends to be overlooked\cite{richard2008modern}. Given the well known duality principle in linear systems \cite{bacciotti2019stability}[Section 5.2.4], \cite{williams2007linear}[Section 4.3] with respect to controllability and observability, the proposed graphical mechanism is applicable to transform from controllability to controller realization. 



An extensive bibliographic search reveals references that briefly elaborate on observability and observer properties of the systems \cite{1971iii,brockett2015finite,heij2006introduction,hespanha2018linear,curfoin1978infinite,farina2011positive,bosgra2001design}, however these references do not explore these concepts in depth. Additional references  introduce the observability matrix from the original system which is essential to construct the observability realization \cite{trentelman2012control,Kalpana2020,antsaklis2006linear,sinha2007linear,gu2012discrete,o1983observers,ControFriedland,williams2007linear,zadeh2008linear,chen1999linear,bourles2013linear,hendricks2008linear,zabczyk2020mathematical,ogata2010modern,brogan1991modern,isidori2013nonlinear,bacciotti2019stability,delchamps2012state}. Finally, other references describe the block diagrams for the observability/observer and controllability/controllable canonical realizations \cite{basile1992controlled,FPE:86,zadeh2008linear,Kai:80,szidarovszky2018linear}, however they rely on the transformation theorem to reach the observer realization from the observability form. The same search indicates that the graphical approach, introduced here, to perform the state-space transformation from observability realization to observer realization, is new. Further investigations about the state-space transformation corresponding to these block manipulations also reveal a close connection between the computation of these state-space transformations and recent results on generalized Fibonacci polynomials.







\section{Preliminaries}

We consider the single-input, single-output system
\beq
\begin{array}{rcl}
\dsp\frac{\dif}{\dif t}\bm{x} & = & \bm{Ax}+\bm{B}u, \; \bm{x}(0)=\bm{x}_0, \\[10pt]
y & = & \bm{Cx}, 
\end{array}
\label{basic_system}
\eeq
where $\bm{x}=\bm{x}(t) \in \reals^{n}$ is the system's state, $u = u(t)$ is the system's input, and $y=y(t)$ is the system's output. As we are addressing the
observability characteristics of the system, for notation clarity purposes, we assume a null input to the system and therefore omit the input matrix $\bm{B}$. However, the reader must be warned that the following discussion makes extensive use of the word ``realization", which is most often understood as applying to systems with both inputs and outputs. Thus, although the input matrix $\bm{B}$ is omitted below, it is always present in the discussions below whenever the word ``realization" is used.

All considerations are relative to the system's observability properties and observer design considerations. These considerations can be transposed to the corresponding controllability and controller design considerations.
A simple criterion to establish the observability of the system~(\ref{basic_system}) is established by the classical observability theorem. 

\begin{theorem}[{See \cite{Kai:80}}]  The system~(\ref{basic_system}) is observable if and only if the square matrix
\[
{\cal O} = \left[\begin{array}{l} \bm{C}\\\bm{CA}\\\bm{CA}^2 \\\hspace{2mm} \vdots \\\bm{CA}^{n-1} \end{array} \right]
\]
is invertible.
\end{theorem} \label{t1}
The observability theorem exposed relates the linear independency of the $n$ rows of $\mathcal{O}$, with the ability to estimate each one of the $n$ states of the system \eqref{basic_system}. The observability matrix, $\mathcal{O}$, is used to extract the {\em observability realization} of the observable system~(\ref{basic_system}).
\begin{theorem}
Consider the basis formed by the vectors $\bm{C}^\top, \; \bm{A}^\top \bm{C}^\top, \; \ldots , {\bm{A}^{n-1}}^\top \bm{C}^\top$. In that basis, the system~(\ref{basic_system}) is realized with the companion form
\beq
\begin{array}{rcl}
\dsp \frac{\dif}{\dif t} \tilde{\bm{x}} & = & 
\left[\begin{array}{ccccc} 0& 1 & 0 &  
 \ldots & 0 \\
0 & 0 & 1 & \ddots & \vdots \\
\vdots & & \ddots & \ddots & \vdots \\
0 & 0 & \ldots & 0 & 1 \\
-a_0 & -a_1 & \ldots & -a_{n-2} & -a_{n-1}
\end{array} \right]\tilde{\bm{x}}\\ [40pt]
& = & \bm{A}_{\rm observability} \tilde{\bm{x}}\\[10pt]
y & = & \left[ 1 \; 0 \; \ldots \; 0 \; 0\right]\tilde{\bm{x}}\\[10pt]
& = & \bm{C}_{\rm observability} \tilde{\bm{x}},
\end{array}
\label{observability_1}
\eeq
with ${\bm{A}}_{\rm observability} = {\cal O}\bm{A}{\cal O}^{-1}$, ${\bm{C}}_{\rm observability} = \bm{C}{\cal O}$, and $a_0 , \dots ,\;a_{n-1}$ are the coefficients of the characteristic polynomial $\lambda(s)$ of $\bm{A}$
\begin{equation}
   \lambda(s)=s^n+a_{n-1}s^{n-1}+\dots+a_1s+a_0.
   \label{eqn:poli}
\end{equation}
\end{theorem}
\begin{proof}
{Proof of Theorem 2 is well know in the literature. See for instance \cite{bacciotti2019stability}[Section 5.2.2], \cite{Kai:80}[Section 6.2] and \cite{antsaklis2006linear} [Section 3.2].}
\end{proof}
While the observability realization~(\ref{observability_1}) and surrounding analyses provide valuable and easy insight into the observability of~(\ref{basic_system}), it fails to provide an immediately constructive mechanism to specify observers in terms of specifying arbitrary poles of observer error dynamics. A realization of~(\ref{basic_system}) that supports this purpose is the observer realization
\beq
\begin{array}{rcl}
\dsp \frac{\dif}{\dif t} \bar{\bm{x}} & = & 
\left[\begin{array}{ccccc} 
-a_{n-1}& 1 & 0 &  \ldots & 0 \\
-a_{n-2} & 0 & 1 & \ddots & \vdots \\
\vdots & & \ddots & \ddots & \vdots \\
-a_1 & 0 & \ldots & 0 & 1 \\
-a_0 & 0 & 0& \ldots & 0
\end{array} \right]\bar{\bm{x}}\\ [40pt]
& = & \bm{A}_{\rm observer} \bar{\bm{x}}\\[10pt]
y & = & \left[ 1 \; 0 \; \ldots \; 0 \; 0\right]\bar{\bm{x}}\\[10pt]
& = & \bm{C}_{\rm observer} \bar{\bm{x}}.
\end{array} 
\label{observer}
\eeq
This realization is convenient for the reader to understand the flexibility offered by linear observers of the form
\[
\begin{array}{rcl}
\dsp \frac{\dif}{\dif t}\hat{\bm{x}} & = & \bm{A}\hat{\bm{x}} + \bm{L}(y - \bm{C}\hat{\bm{x}})\\
&=& (\bm{A}-\bm{LC})\hat{\bm{x}}+ \bm{L}y.
\end{array}
\]
This kind of observer is known as Luenberger observer, where $\bm{L} \in \reals^{n}$ is the observer gain.
Indeed, in the transformed state-space that defines $\bar{x}$, 
the observer becomes
\[
\begin{array}{rcl}
\dsp \frac{\dif}{\dif t}\hat{\bar{\bm{x}}} & = & \bm{A}_{\rm observer}\hat{\bar{\bm{x}}} + \bm{L}_{\rm observer}(y - \bm{C}_{\rm observer}\hat{\bar{\bm{x}}})\\[10pt]
&=& (\bm{A}_{\rm observer}-\bm{L}_{\rm observer}\bm{C}_{\rm observer})\hat{\bar{\bm{x}}}+ \bm{L}_{\rm observer}y\\[10pt]
&=& \left[\begin{array}{ccccc} 
-a_{n-1}-l_{n-1}& 1 & 0 &  \ldots & 0 \\
-a_{n-2}-l_{n-2} & 0 & 1 & \ddots & \vdots \\
\vdots & & \ddots & \ddots & \vdots \\
-a_1-l_1 & 0 & \ldots & 0 & 1 \\
-a_0 -l_0& 0 & 0& \ldots & 0
\end{array} \right]\hat{\bar{\bm{x}}} + \left[\begin{array}{l} l_{n-1} \\ l_{n-2} \\ \hspace{1mm}\vdots \\ l_1 \\ l_0 \end{array} \right]y.
\end{array}
\]
This realization clearly shows how the observer gains can be used to modify the characteristic polynomial of the estimation error dynamics arbitrarily since each gain directly and simply impacts each corresponding characteristic polynomial coefficient. 
The derivation of the observer realization~(\ref{observer}) of the 
system~(\ref{basic_system}) from its observability realization  (\ref{observability_1}) is, however, not immediate and existing derivations, which rely on a number of linear-algebraic constructs, can constitute an obstacle to the senior undergraduate or first-year graduate students who learns this topic. 

We have found a way to go from Observability to Observer realization in the current literature. It relies on the transfer function for the derivation of the observer realization form. This is immediate but relies on having the system's transfer function. We propose two alternative derivations of the observer realization via matrix manipulations and elementary block manipulations that follow the spirit of those found in Kailath~\cite[pp. 37-45]{Kai:80}.

\section{Going from Observability to Observer realizations via Matrix Manipulations}

Given the system \eqref{basic_system}, a linear state-space transformation is given by any invertible matrix $\bm{P}\in\reals^{n\times n}$ such that $\hat{\bar{\bm{x}}}=\bm{Px}$, resulting in the new realization

\begin{align}
\frac{\dif }{\dif t}\hat{\bar{\bm{x}}}&=\bm{PAP}^{-1}\hat{\bar{\bm{x}}}\\
y&=\bm{CP}^{-1}\hat{\bar{\bm{x}}}.
    \label{eqn:transformationstates}
\end{align}

Considering that $\bm{A}_{\rm{observer}}=\bm{PA}_{\rm{observability}}\bm{P}^{-1}$, acceptable linear state-space transformations may be found by solving for \eqref{eqn:suylvester} and \eqref{eqn:suylvester2}.

\begin{align}
    \bm{A}_{\rm{observer}}\bm{P}-\bm{PA}_{\rm{observability}}=0,\label{eqn:suylvester}\\
    \bm{C}_{\rm{observer}}\bm{P}-\bm{C}_{\rm{observability}}=0.
    \label{eqn:suylvester2}
\end{align}
The complete transformation of the realization also involves the condition $\bm{B}_{\rm observer}- \bm{PB}_{\rm observability} = 0$. However, this condition does not participate actively to the rest of the paper and is therefore omitted.

Equation \eqref{eqn:suylvester} is a homogeneous Sylvester equation in the unknown $\bm{P}$. The equivalence of $\bm{A}_{\rm{observer}}$ and $\bm{A}_{\rm{observability}}$ implies that there exists an invertible matrix $\bm{P}$ that satisfies \eqref{eqn:suylvester} and \eqref{eqn:suylvester2}. One such solution is


\begin{equation*}
\bm{P}=\left[\begin{array}{cccc}
1 & 0 & \cdots & 0 \\
a_{n-1} & 1 & \ddots & \vdots \\
\vdots & \ddots & \ddots & 0 \\
a_{0} & \cdots & a_{n-1} & 1
\end{array}\right].
    \label{eqn:P}
\end{equation*}

$\bm{P}$ is a Toeplitz matrix of dimension $n\times n$. It is remarkable, that as it was expected all the eigenvalues of the system are equal to one, given that it means that the spectrum is not changed by doing this transformation. The matrix $\bm{P}^{-1}$ is characterized by \cite{Sahin2018}

\begin{align}
\bm{P}^{-1}=\left[\begin{array}{cccc}
1 & 0 & \cdots & 0 \\
F_1 & 1 & \ddots & \vdots \\
\vdots & \ddots & \ddots & 0 \\
F_{n-1} & \cdots & F_1 & 1
\end{array}\right],
    \label{eqn:Pinv}
\end{align}
where $F_k$ is the generalized Fibonacci polynomial of order $k$ on the variables $(a_{n-k},a_{n-k+1},a_{n-k+2},\dots,a_{n-1})$. The generalized Fibonacci polynomials ($F_k$), with $k \in \{0, \ldots, n-1\}$, are defined inductively by

$$
\begin{aligned}
F_{0} &=1 \\
F_{1} &=-a_{n-1}F_0 \\
F_{2} &=-a_{n-1} F_{1}-a_{n-2} F_{0} \\
F_{3} &=-a_{n-1} F_{2}-a_{n-2} F_{1}-a_{n-3}F_0 \\
\vdots\\
F_k&=\sum_{i=1}^{k}-a_{n-i}F_{k-i}\\
\vdots\\
F_{n-1} &=-a_{n-1} F_{n-2}-a_{n-2}F_{n-3}-\cdots-a_{2} F_{1}-a_1F_0.
\end{aligned}
$$

Each $F_k$ may also be calculated as the determinant of the lower Hessenberg matrix $\bm{H}_k$ \cite{Sahin2018}, with 

\begin{equation*}
    \bm{H}_{k}=\left[\begin{array}{ccccc}
a_{n-1} & -1 & 0 & \cdots & 0 \\ 
a_{n-2} & a_{n-1} & -1 & \cdots & 0 \\
\vdots & \vdots & \vdots & \ddots & \vdots \\
a_{n-k+1} & a_{n-k+2} & a_{n-k+3} & \cdots & -1 \\
a_{n-k} & a_{n-k+1} & a_{k-2} & \cdots & a_{n-1}
\end{array}\right].
\label{eqn:Hmatrix}
\end{equation*}

Although interesting in its own right, and apparently new to the control community, the explicit computation of $\bm{P}$ and its inverse $\bm{P}^{-1}$, the transformation matrix allowing to go from observability to observer realizations is heavy with computations and not very friendly to the novice. This is why we introduce a \textit{mechanical and graphical process} to do the same by means of a sequence of elementary block manipulations.

\section{Transforming the observability realization into the observer realization via elementary block manipulations}

We begin with the observability realization~(\ref{observability_1}), which we represent graphically in Fig.~\ref{observability}.

\begin{table}[]
\centering
\begin{tabularx}{\columnwidth}{|X|}
\hline
\includegraphics[width=\linewidth]{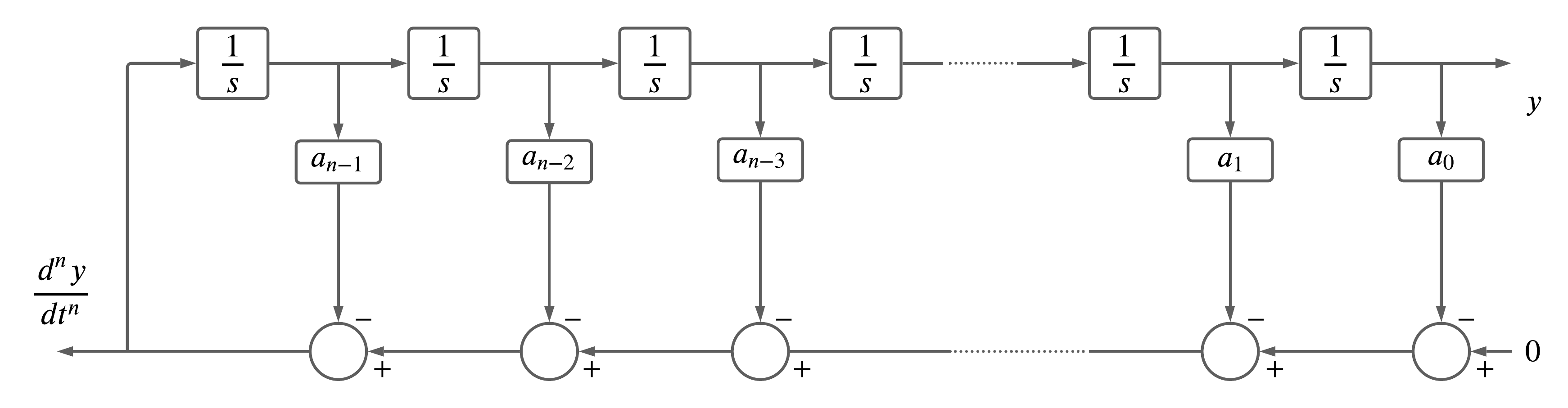}
\captionof{figure}{Block-diagram representation of the observability realization of the system~(\ref{basic_system}).}\label{observability} \\ \hline
\begin{equation}
y^{(n)}+y^{(n-1)}a_{n-1}+\dots+a_0y=0.
\label{system}
\end{equation} \\ \hline
\end{tabularx}
\caption{Block-diagram representation of the observability realization of the system~(\ref{basic_system}) with the output $y$ ODE.}
\label{tab:t1}
\end{table}


The state-space transformation to the observer realization is obtained by moving integrators around the block diagram.\ 
With this proposed method, the system is transformed into several observable and non-observable realizations in the described steps.

The first move is shown in Fig.~\ref{move_2}, whereby the upper left integrator is moved down, replacing $\dif^n y/\dif t^n$ by $\dif^{n-1}y/\dif t^{n-1}$ on the way. 


\begin{table}[]
\centering
\begin{tabularx}{\columnwidth}{|X|}
\hline
\includegraphics[width=\linewidth]{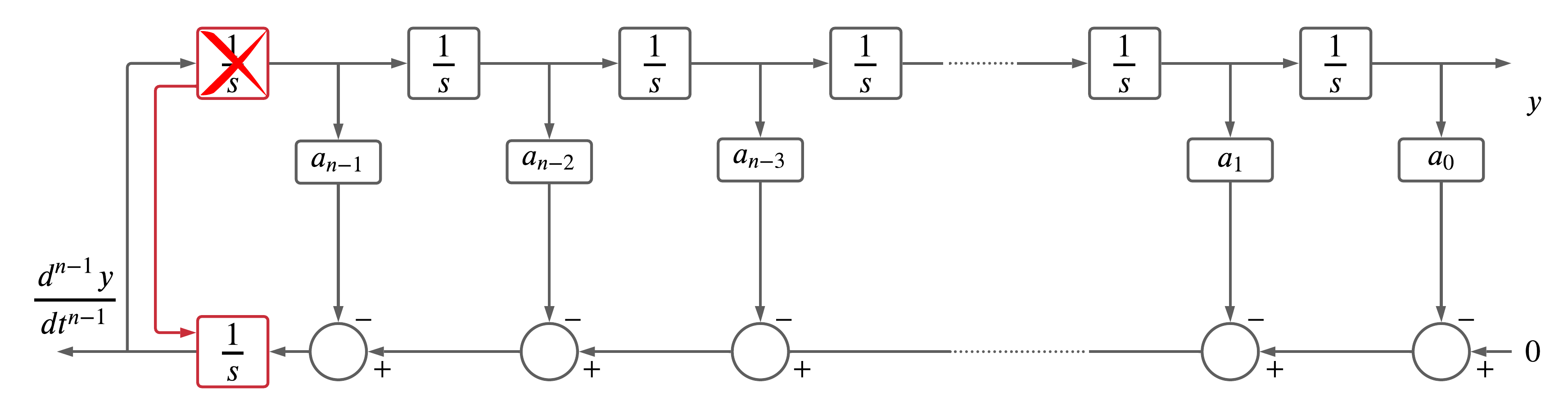}
\captionof{figure}{First move: The upper left integrator is moved down.}\label{move_2} \\ \hline
\begin{equation*}
y^{(n-1)}=\textcolor{red}{\int}y^{(n-1)}(-a_{n-1})\textcolor{red}{}+\textcolor{red}{\int}\left[\sum_{i=1}^{n-1}(-a_{i-1})\int^{(n-i)}y^{(n-1)}\right]\textcolor{red}{}
\end{equation*}
\begin{equation*}
y^{(n-1)}=y^{(n-2)}(-a_{n-1})+\textcolor{red}{\int}\left[\sum_{i=1}^{n-1}(-a_{i-1})y^{(i-1)}\right]\textcolor{red}{}
\end{equation*}
\begin{equation*}
y^{(n-1)}=y^{(n-2)}(-a_{n-1})+\sum_{i=1}^{n-1}(-a_{i-1})y^{(i-2)}
\end{equation*}
\begin{equation*}
y^{(n-1)}=\sum_{i=1}^n(-a_{i-1})y^{(i-2)}
\end{equation*}
So, we obtain an integrated version of ODE \ref{system} in TABLE \ref{tab:t1}:
\begin{equation*}
y^{(n-1)}+a_{n-1}y^{(n-2)}+\dots+a_1y+a_0\int y=0.
\end{equation*}
 \\ \hline
\end{tabularx}
\caption{First move: The upper left integrator is moved down. We also observe the associated ODE.}
\label{tab:t2}
\end{table}


The movement of the same integrator behind the sum splits it due to linearity. Also, the following upper left integrator is moved down, where a differentiator is created next to the gain $a_{n-1}$. We observe these two steps in
 Fig.~\ref{move_3}.

\begin{table}[]
\centering
\begin{tabularx}{\columnwidth}{|X|}
\hline
\includegraphics[width=\linewidth]{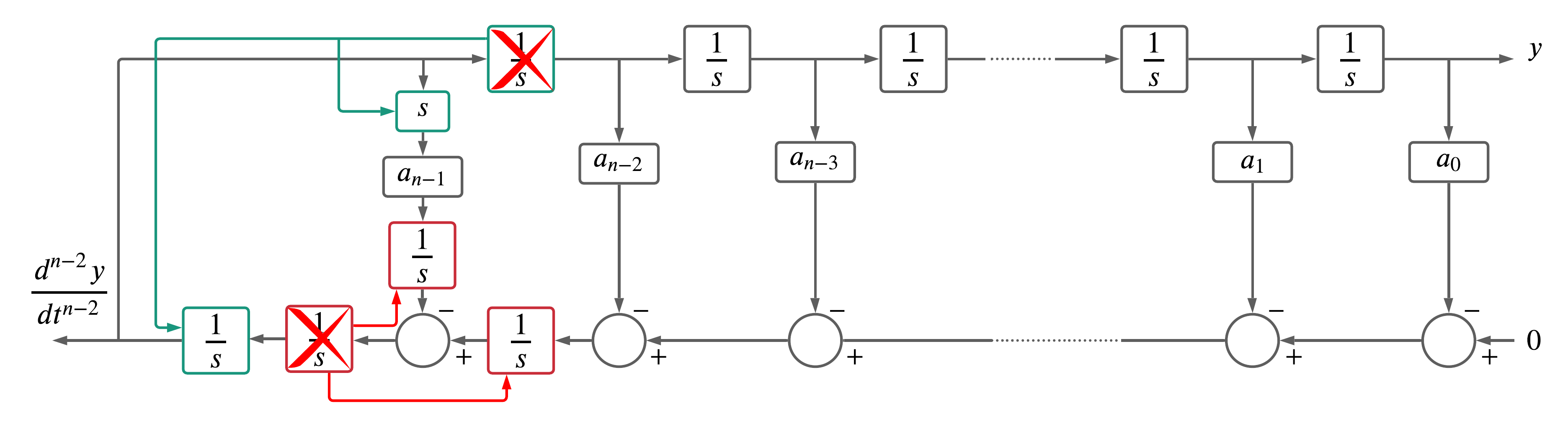}
\captionof{figure}{Second and third move: The lower left integrator is moved before the sum, where it is split due to linearity (second move). The next upper left integrator is moved down, where a differentiator is created next to the gain $a_{n-1}$ (third move).}\label{move_3} \\ \hline
\begin{multline*}
y^{(n-2)}=\textcolor{green}{\int}\textcolor{red}{\int}\textcolor{green}{\frac{\dif}{\dif t}}\left((-a_{n-1})y^{(n-2)}\right)\textcolor{red}{}\textcolor{green}{}\\
+{\color{green}\int}{\color{red}\int}\left((-a_{n-2})y^{(n-2)}\right){\color{red}}{\color{green}}+{\color{green}\int}{\color{red}\int}\left(\sum_{i=3}^{n}\int^{(i-2)}(-a_{i-1})y^{(n-1)}\right)\textcolor{red}{}\textcolor{green}{}
\end{multline*}
\begin{equation*}
y^{(n-2)}=(-a_{n-1})y^{(n-3)}+(-a_{n-2})y^{(n-4)}+\sum_{i=3}^n(-a_{n-i})y^{(n-2-i)}
\end{equation*}
\begin{equation*}
y^{(n-2)}=\sum_{i=1}^n(-a_{n-i})y^{(n-2-i)}.
\end{equation*}
So, we obtain a two-times integrated version of ODE \ref{system}:
\begin{equation*}
y^{(n-2)}+a_{n-1}y^{(n-3)}+\dots+a_1\int y+a_0\int\int y =0.
\end{equation*}
 \\ \hline
\end{tabularx}
\caption{Second and third move: The lower left integrator is moved before the sum, where it is split due to linearity (second move). The next upper left integrator is moved down, where a differentiator is created next to the gain $a_{n-1}$ (third move). We also observe the associated ODE.}
\label{tab:t3}
\end{table}


Since blocks $s$ and $1/s$ next to gain $a_{n-1}$ represent inverse linear operators, we cancel them and we reach the realization in Fig.~\ref{move_4}.

\begin{table}[]
\centering
\begin{tabularx}{\columnwidth}{|X|}
\hline
\includegraphics[width=\linewidth]{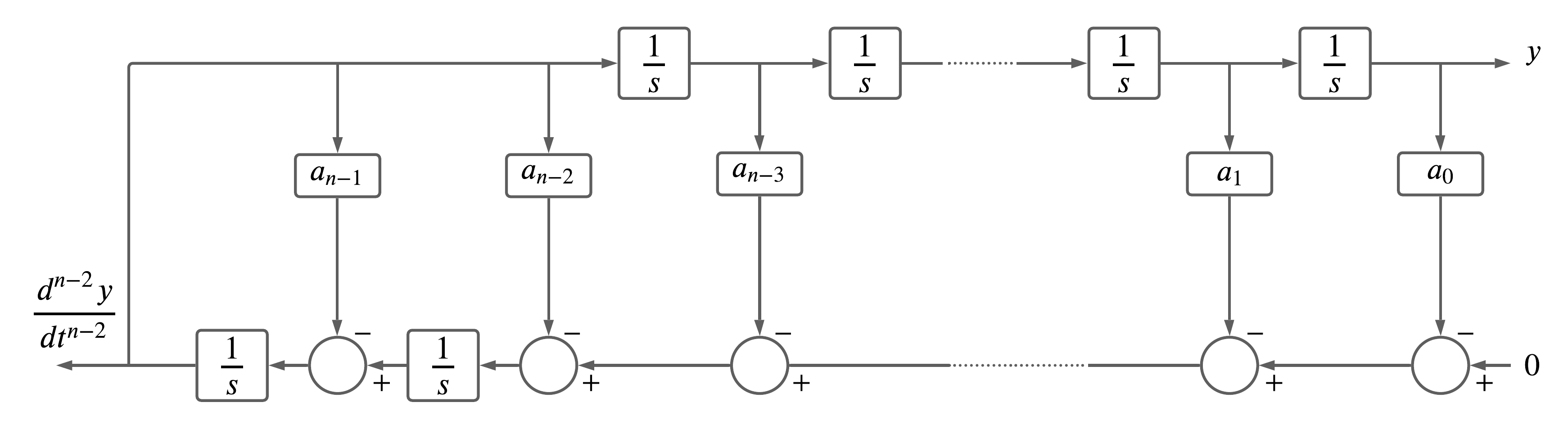}
\captionof{figure}{Realization after canceling integrator and differentiator next to gain $a_{n-1}$.}\label{move_4} \\ \hline
\begin{multline*}
y^{(n-2)}=\int(-a_{n-1})y^{(n-2)}+\int\int(-a_{n-2})y^{(n-2)}\\
+\int\int\left(\sum_{i=3}^n\left(-a_{n-i}\int^{(i-2)}y^{(n-2)}\right)\right)
\end{multline*}
\begin{equation*}
y^{(n-2)}=\sum_{i=1}^n\int^{(i)}(-a_{n-i})y^{(n-1)}.
\end{equation*}
So, we obtain a two-times integrated version of ODE \ref{system}:
\begin{equation*}
y^{(n-2)}+a_{n-1}y^{(n-3)}+\dots+a_1\int y+a_0\int\int y =0.
\end{equation*}
 \\ \hline
\end{tabularx}
\caption{Realization and ODE for the output after canceling the integrator and differentiator next to gain $a_{n-1}$.}
\label{tab:t4}
\end{table}

We keep repeating the preciously described process a total of $m$ times, where $m$ integrator are down, and $n-m$ are up. Also, there are no integrators nor differentiators next to gains $a_{n-1},a_{n-2},\dots,a_{n-(m-1)}$ due to cancelation. In Fig. \ref{move_7} we observe the associated realization.

\begin{table}[]
\centering
\begin{tabularx}{\columnwidth}{|X|}
\hline
\includegraphics[width=1\textwidth]{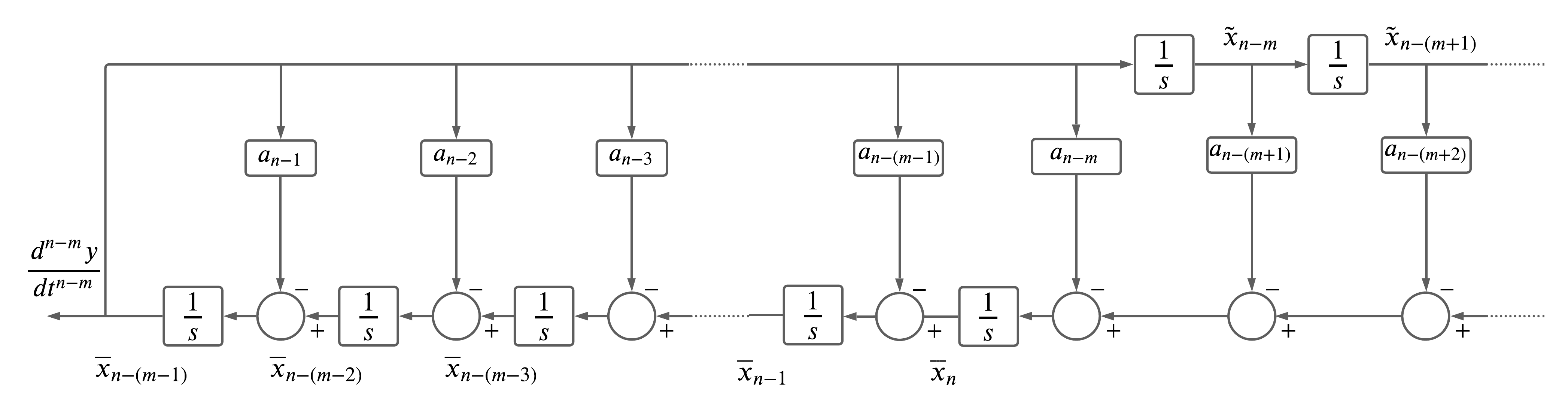}
\captionof{figure}{We observe the realization after moving down and aligning the first $m$ integrators.}\label{move_7} \\ \hline
\begin{equation*}
y^{(n-m)}=\sum_{i=1}^m(-a_{n-i})\int^{(i)}y^{(n-m)}+\int^{(m)}\sum_{i=1}^{n-m}\int^{(i)}(-a_{n-m-i})y^{(n-m)}
\end{equation*}
\begin{equation*}
y^{(n-m)}=\sum_{i=1}^n\int^{(i)}(-a_{n-i})y^{(n-m)}.
\end{equation*}
So, we obtain a $m$-times integrated version of ODE \ref{system}:
\begin{equation*}
y^{(n-m)}+a_{n-1}y^{(n-m-1)}+\dots+a_1\int^{(m-1)}y+a_0\int^{(m)} y=0.
\end{equation*}
 \\ \hline
\end{tabularx}
\caption{Realization, state-space matrix, and output ODE for the system after moving down $m$ integrators.}
\label{tab:t7}
\end{table}


It is now possible to repeat the foregoing steps {shown in the Fig.~\ref{move_2}, until Fig.~\ref{move_7}} and progressively move all remaining upper integrators until none remains, resulting in the observer realization shown in Fig.~\ref{observer_fig}.


\begin{table}[]
\centering
\begin{tabularx}{\columnwidth}{|X|}
\hline
\includegraphics[width=1\textwidth]{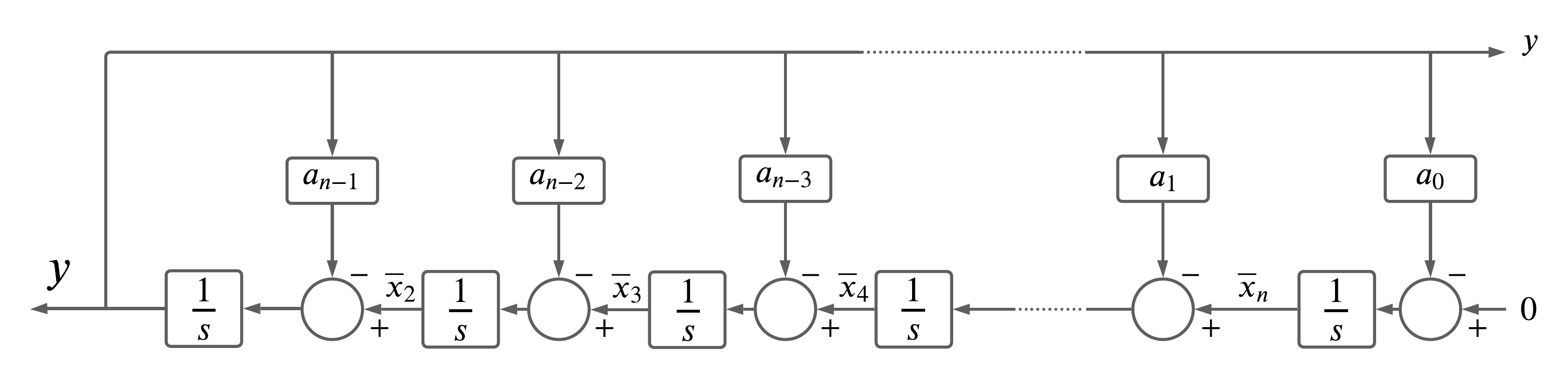}
\captionof{figure}{Observer  realization.}\label{observer_fig} \\ \hline
\begin{equation}
y+a_{n-1}\int y+a_{n-2}\int\int y+\dots+a_0\int^{(n)}y=0.
\label{fig:final_ode}
\end{equation}
 \\ \hline
\end{tabularx}
\caption{Observer realization.}
\label{tab:observer_fig}
\end{table}

In the next section we analyze the block manipulations as a succession of elementary state-space transformations.

Before leaving this section, we would like to thank the reviewer for suggesting to draw the equivalence between the suggested manipulations and the progressive integration of the differential equation \eqref{system} to yield the integral equation \eqref{fig:final_ode}. We utilize $\int^{(n)}y$ to denote the process of integrating $n$ times the function $y$ with respect to time.


\section*{State-space interpretation of elementary block manipulations}

{ We call $\bm{A}_m$ to the state-matrix associated to the $m$-th observable realization during the block manipulation. We call \textit{block manipulation} to the process of going from Table \ref{tab:t7} for step $m$ to Table \ref{tab:t7} for step $m+1$. Essentially, we start with $\bm{A}_{\rm{observability}}$, and we move through the sequence $\bm{A}_1,\dots,\bm{A}_{n-1}$, until we reach $\bm{A}_{\rm{observer}}$. $\bm{A}_m$ is given by

{\small
\begin{equation*}
\bm{A}_m= 
\begin{bmatrix}
\textbf{0} & 1 & 0 & \dots & \dots & \dots & \dots & \dots & \dots & \dots & \dots & 0 & 0 & 0 \\
0 & \textbf{0} & 1 & 0 & \dots & \dots & \dots & \dots & \dots & \dots & \dots & 0 & 0 & 0 \\
\dots & \dots & \dots & \dots & \dots & \dots & \dots & \dots & \dots & \dots & \dots & \dots & \dots & \dots \\
0 & 0 & 0 & \dots & \textbf{0} & 1 & 0 & \dots & \dots & \dots & \dots & 0 & 0 & 0 \\
0 & 0 & 0 & \dots & \dots & \textbf{0} & 1 & 0 & \dots & \dots & \dots & 0 & 0 & 0 \\
0 & 0 & 0 & \dots & \dots & \dots & \bm{-a_{n-1}} & 1 & 0 & \dots & \dots & 0 & 0 & 0 \\
0 & 0 & 0 & \dots & \dots & \dots & -a_{n-2} & \textbf{0} & 1 & 0 & \dots & 0 & 0 & 0 \\
\dots & \dots & \dots & \dots & \dots & \dots & \dots & \dots & \dots & \dots & \dots & \dots & \dots & \dots \\
0 & 0 & 0 & \dots & \dots & \dots & -a_{n-(m+1)} & \dots & \dots & \dots & \dots & 0 & \textbf{0} & 1 \\
-a_{0} & -a_{1} & -a_{2} & \dots & \dots & -a_{n-(m-1)} & -a_{n-m} & 0 & \dots & \dots & \dots & \dots & \dots & \textbf{0} \\
\end{bmatrix}.
\end{equation*} }
Since all observable realizations have the same dynamics, i.e., they are equivalent. There exists a change-of-basis matrix $\bm{P}_m$ such that $\bm{A}_m=\bm{P}_m\bm{A}_{m-1}\bm{P}^{-1}_m$ leading to

\begin{equation}
\bm{P}_{n-1}... \bm{P}_{1} \bm{A}_{\rm{observability}}\bm{P}^{-1}_1... \bm{P}^{-1}_{n-1}=\bm{A}_{\rm{observer}}.
    \label{eqn:transitions}
\end{equation}
The change-of-basis matrix $\bm{P}_i$ is given by
\begin{equation*}
    \bm{P}_i=\begin{bmatrix}
1 & 0 & 0  & \dots & \dots & \dots & \dots & 0 & 0 \\
0 &\ddots & \ddots & & & & & & 0 \\
\vdots & \ddots & 1 & \ddots &  &  &  &  & \vdots \\
\vdots &  & 0 & 1 & \ddots & &  & & \vdots \\
\vdots & & \vdots & -a_{n-1} & 1 & \ddots &  & & \vdots \\
\vdots & & \vdots & -a_{n-2} & 0 & \ddots & \ddots & & \vdots \\
\vdots & & \vdots & \vdots & \vdots & \ddots & \ddots & \ddots & 0 \\
\vdots & & \vdots & -a_{n-(i-1)} & \vdots & & \ddots & \ddots & 0 \\
 0 & \dots & 0 & -a_{n-i} & 0 & \dots & \dots & 0 & 1 \\
\end{bmatrix},
\end{equation*}
and its inverse $\bm{P}_i^{-1}$ by
\begin{equation*}
    P_i^{-1}=\begin{bmatrix}
1 & 0 & 0  & \dots & \dots & \dots & \dots & 0 & 0 \\
0 &\ddots & \ddots & & & & & & 0 \\
\vdots & \ddots & 1 & \ddots &  &  &  &  & \vdots \\
\vdots &  & 0 & 1 & \ddots & &  & & \vdots \\
\vdots & & \vdots & a_{n-1} & 1 & \ddots &  & & \vdots \\
\vdots & & \vdots & a_{n-2} & 0 & \ddots & \ddots & & \vdots \\
\vdots & & \vdots & \vdots & \vdots & \ddots & \ddots & \ddots & 0 \\
\vdots & & \vdots & a_{n-(i-1)} & \vdots & & \ddots & \ddots & 0 \\
 0 & \dots & 0 & a_{n-i} & 0 & \dots & \dots & 0 & 1 \\
\end{bmatrix},
\end{equation*}
}%
where the coefficients $a_{n-1},\ldots,a_{n-i}$ are placed on the column $n-i$. Naturally, \eqref{eqn:transitions} leads to the matrices $\bm{P}$ and $\bm{P}^{-1}$ in \eqref{eqn:P} and \eqref{eqn:Pinv} respectively, as

\begin{align*}
    \prod_{i=1}^{n-1}\bm{P}_{n-i}&=\bm{P}\\ \prod_{i=1}^{n-1}\bm{P}^{-1}_{i}&=\bm{P}^{-1}.
\end{align*}

Furthermore, given that $\bm{C}_{m}=\bm{C}_{m-1} \bm{P}^{-1}_{m}$, it is possible to see that $\bm{C}_{\rm{observability}}=\bm{C}_m=\bm{C}_{\rm{observer}}$.

Despite the intrinsic value of these elementary state-space transformations, the authors hope that the reader will agree with them when they say the equivalent block-diagram manipulations are far more accessible to the senior undergraduate or a first-year graduate student interested in learning about linear systems and navigating their canonical state-space representations.

\section*{Conclusion}
A visual and easy-to-understand transformation from a LTI system in observability form to observer form is introduced. This transformation solely relies on elementary block diagram manipulations, thus making it accessible to students during their early exposure to state-space systems.
The third author author has made this transformation a part of his elementary linear systems theory curriculum. It also applies to controllability and controller realizations via duality.

\processdelayedfloats 

\sidebars 

\clearpage
\newpage








\newpage
\processdelayedfloats 
\clearpage






\newpage
\section{Author Biography}

Eder Baron-Prada received a B.S. degree in Electrical Engineering and the M.S. degree in Industrial Automation from the Universidad Nacional de Colombia, Bogot\'a, Colombia, in 2016 and 2019 respectively. He is currently a Ph.D. student at King Abdullah University of Science and Technology (KAUST) in Electrical and Computer Engineering. His research interests include game theory, optimization, and control of multi-agent systems applied to modern power systems.

Renzo Caballero received an Electrical Engineering degree from Universidad de la Rep\'{u}blica, Uruguay, in 2016, and an M.S. in Applied Mathematics and Computer Science degree from King Abdullah University of Science and Technology (KAUST), KSA, in 2019. He is currently pursuing his PhD at KAUST. His current research focuses on self-replicating robots, but he also is interested in microelectronic design, synthesis of 2D materials, and stochastic optimal control.

Eric Feron is a professor of Electrical, Computer, and Mechanical Engineering. He is the director of the Robotics, Intelligent Systems, and Control (RISC) Laboratory. He recently joined the KAUST CEMSE Division from the Georgia Institute of Technology. Prior to his time at Georgia Tech, he was an active faculty member in MIT's Aeronautics and Astronautics department from 1993 until 2005. Feron's career in academia began in Paris, France, where he obtained the B.S. and M.S. from \'Ecole Polytechnique and \'Ecole Normale Sup\'erieure, respectively. He later completed the Ph.D. in aerospace engineering at Stanford University, USA.

\newpage

\bibliographystyle{IEEEtran}
\bibliography{bibliography.bib}

\end{document}